\RequirePackage{fix-cm}
\documentclass[smallextended]{svjour3}       

\smartqed  
\usepackage{graphicx}

\usepackage{color}

\usepackage{fancybox}

\definecolor{aliceblue}{rgb}{0.94, 0.97, 1.0}

\usepackage{graphicx}

\usepackage{xcolor}
\usepackage{fancybox}

\usepackage[most]{tcolorbox}

\usepackage{latexsym}

\usepackage{amsfonts,amsmath}

\usepackage{amssymb}

\usepackage[mathscr]{eucal}

\begin{document}

\title{
An homogenization approach for the inverse spectral problem of periodic Schr\"odinger operators
}

\titlerunning{An homogenization approach for the inverse spectral problem}        

\author{
        Lorenzo Zanelli 
}


\institute{
           L. Zanelli \at
           Department of Mathematics ``Tullio Levi-Civita'', University of Padova\\
            \email{lzanelli@math.unipd.it}
}

\date{Received: date / Accepted: date}

\maketitle

\begin{abstract}
We study the inverse spectral problem for periodic Schr\"odinger opera\-tors of kind $- \frac{1}{2} \hbar^2 \Delta_x  +   V(x)$ on the flat torus $\Bbb T^n :=  (\Bbb R / 2 \pi \Bbb Z)^n$  with potentials $V \in C^{\infty} (\Bbb T^n)$. We show that if two operators are isospectral for any $0 < \hbar \le 1$ then they have the same effective Hamiltonian given by the  periodic homogenization of  Hamilton-Jacobi equation. This result provides a necessary condition for the isospectrality of these  Schr\"odinger operators. We also provide a link between our result and    the spectral limit of quantum integrable systems.

\end{abstract}

\section{Introduction}
\label{intro}
Let $\Bbb T^n :=  (\Bbb R / 2 \pi \Bbb Z)^n$, $V \in C^\infty (\Bbb T^n)$ and $H (x,p) := \frac{1}{2} |p|^2 +  V(x)$.  The related class of Schr\"odinger operators   is given by 
\begin{equation}
\label{Sch1}
{\rm Op}_\hbar^w  (H)  = - \frac{1}{2}   \hbar^2 \Delta_x  +   V(x). 
\end{equation} 
This operator is selfadjoint on $W^{2,2} (\Bbb T^n)$ and exhibits discrete spectrum which is bounded from below  (see for example \cite{R-S}).
The semiclassical inverse spectral problem is the study of the family of those $H$ such that ${\rm Spec}({\rm Op}_\hbar  (H) )$ is the same for all $0 < \hbar \le 1$, namely with the same eigenvalues and the same multiplicity.\\
\indent The target of our paper is to discuss the link between this semiclas\-sical inverse spectral problem for ${\rm Op}^w_\hbar  (H)$ and the periodic homogenization of the Hamilton-Jacobi equation of  $H$ (see for example \cite{E2}, \cite{E4}, \cite{E5}, \cite{LHY}, \cite{Ri}).

The first observation we provide is about the Egorov Theorem for the propagation of quantum observables (see  \cite{B-G-P}, \cite{B-R} and the references therein). For classical observables $A :  \Bbb R^{2n}_z \rightarrow \Bbb R$ given by  symbols  associated to  Weyl operators ${\rm Op}^w_\hbar  (A)$ and for symbols $B :  \Bbb R^{2n}_z \rightarrow \Bbb R$, $|\partial_z^\gamma B(z)| \le C_\gamma$ $\forall \, |\gamma| \ge 2$   with Hamiltonian flow  $\varphi^t : \Bbb R^{2n}  \to \Bbb R^{2n}$, the Egorov Theorem  gives
\begin{equation}
\label{Ego}
U_\hbar^\star (t) \circ   {\rm Op}_\hbar^w  (A)  \circ  U_\hbar  (t)  =   {\rm Op}_\hbar^w  (A_\varphi)
\end{equation} 
by the quantum dynamics $U_\hbar (t) := \exp{ (- i \, {\rm Op}_\hbar^w (B) t / \hbar)}$ and where 
\begin{equation}
\label{Ego-asymp}
A_\varphi \sim A \circ \varphi^t + \mathcal{O}(\hbar)
\end{equation}
is the semiclassical asymptotics providing the leading term in the space of symbols on $\Bbb R^{2n}_z$.
In Section \ref{SEC2} we recover (\ref{Ego}) and (\ref{Ego-asymp}) in the setting of the Weyl quantization on the flat torus.\\ 
 \indent
 The second observation we underline is a  picture of classical mechanics.  It is well known that,  without integrability or KAM assumptions on $H$, the related Hamilton-Jacobi equation  has not global and smooth solutions.  However, we can consider (see Section \ref{SEC4}) the  viscosity solutions $S(P,\cdot \,) : \Bbb T^n \to \Bbb R$ of
\begin{equation}
\label{HJ00}
H (x,P+\nabla_x S(P,x)) = \overline{H} (P),  \quad P \in \Bbb R^n.
\end{equation}
The function $\overline{H}$, called effective Hamiltonian, is one of the main outcomes  of weak KAM theory and homogenization  of Hamilton-Jacobi equation. It  
 is a convex function and can be represented or approximated in various ways (see for example \cite{E2},  \cite{E4},  \cite{E5}). 
Here we are interested in a symplectic invariance property of $\overline{H}$.   As shown in \cite{B-S},  for  all the time one Hamiltonian flows  $\varphi \equiv \varphi^1 : \Bbb T^n \times \Bbb R^n \to \Bbb T^n \times \Bbb R^n$ with $C^1$ - regularity, i.e.  Hamiltonian diffeomorphisms,   we have the invariance 
\begin{equation}
\label{inv-sym}
\overline{H \circ \varphi }   =  \overline{H} .
\end{equation}
Together with this property, we can to look at the  inverse homogenization problem. This is given by looking at the family of those $H \in C^\infty (\Bbb T^n  \times \Bbb R^n)$ such that the effective Hamiltonian $\overline{H}$ is the same. Such a  problem has been recently studied  in \cite{LHY}, where a characterization in the  case $n=1$ is given and various results for $n\ge1$ are provided.\\
In the Section 4 of \cite{LHY} it is shown a result on the link between this inverse homogenization problem and the spectrum of the Hill operator $- \frac{1}{2} \frac{d^2}{dx^2}  +   V(x)$, namely when $n=1$ and $\hbar = 1$. In this setting, it is proved that the isospectrality implies the same effective Hamiltonian. Our paper aims to provide a result on this link in higher dimensions, but instead of  setting $\hbar = 1$ we will assume the stronger isospectrality for all $0 < \hbar \le 1$. 

In view of the above two observations, we will make use of (\ref{Ego}), (\ref{Ego-asymp})  and (\ref{inv-sym}) in order to prove the two main results of the paper. More precisely, we notice that the isospectrality of two Schr\"odinger operators of kind (\ref{Sch1}), which exhibit descrete spectrum, implies their conjugation by a unitary operator $U_\hbar : L^2 (\Bbb T^n) \to L^2 (\Bbb T^n)$, 
\begin{equation}
\label{Ego6}
U_\hbar^\star  \circ   {\rm Op}_\hbar^w  (H_1)  \circ  U_\hbar    =   {\rm Op}_\hbar^w  (H_2)   . 
\end{equation}
In the framework of our paper, we can prove (see Lemma AB) that this unitary operator takes the form  $U_\hbar = \exp{ (- i \, {\rm Op}_\hbar^w (b (\hbar)) / \hbar)} + \mathcal{O}_{L^2 \to L^2} (\hbar^\infty)$. 
The properties of the symbol $b(\hbar, \cdot \,): \Bbb T^n_x \times \Bbb R^n_p \to \Bbb C$ and its time-one Hamiltonian flow $\varphi_\hbar : \Bbb T^n_x \times \Bbb R^n_p \to \Bbb T^n_x \times \Bbb R^n_p$ do not allow a global  (on $\Bbb T^n \times \Bbb R^n$) semiclassical asymptotics of $H_2 = H_1 \circ \varphi_\hbar +  \mathcal{O}(\hbar)$ in the space of symbols as shown in (\ref{Ego-asymp}). Nevertheless, it is possible a local $C^0$ - semiclassical asymptotics and  this is the content of the first result of the paper   
\begin{theorem}
\label{TH1}
Let $V_\alpha \in C^\infty (\Bbb T^n)$ and $H_\alpha (x,p) := \frac{1}{2} |p|^2 +  V_\alpha(x)$ with $\alpha = 1,2$.   We assume that 
\begin{equation}
\label{spec12}
{\rm Spec}(  {\rm Op}_\hbar^w (H_1) )  =  {\rm Spec}(  {\rm Op}_\hbar^w (H_2)  )  \quad  \quad \forall \, 0 < \hbar \le 1.
\end{equation}
Let $E > \inf_{x \in \Bbb T^n} V_2 (x)$ and  
$
\Omega (E) :=   \{ (x,p) \in \Bbb T^n \times \Bbb R^n \ | \  H_2 (x,p) < E    \}.
$
Then, for $0 < \varepsilon \le 1$ there exist $\sigma_0 (E,\varepsilon) \to 0^+$ as $\varepsilon \to 0^+$  and a family of Hamiltonian diffeomorphisms $\varphi_{\sigma} : \Bbb T^n \times \Bbb R^n \to \Bbb T^n \times \Bbb R^n$ such that 
\begin{equation}
\label{A12}
\|  H_1 \circ \varphi_{\sigma}   -  H_2   \|_{C^0 (\Omega(E)) }  \le   \varepsilon,  \quad  \quad \forall \, 0 < \sigma \le \sigma_0  (E,\varepsilon).
\end{equation}  
\end{theorem}
This local asymptotics does not guarantee  the link of $H_1$ and $H_2$ by a unique Hamiltonian diffeomorphism on $\Omega (E)$ without a remainder. 
Recalling Thm. 1 of \cite{Cr}, if  the unitary conjugation in the lefthand side of (\ref{Ego6}) is an order preserving algebra isomorphism of semiclassical Pseudodifferential operators filtered by powers of $\hbar$  then  there exists  a symplectic diffeomorphism $\psi$ on $\Bbb T^n \times \Bbb R^n$ such that $H_2  =  H_1 \circ \psi$. Unluckily, such a condition is not easy to check in our case since the symbol $b(\hbar, z)$ with $z=(x,p)$ does not necessarily admits the semiclassical asymptotics $b(\hbar, z ) \sim \sum_{j \ge 0} \hbar^j b_j ( z)$.\\
\indent The inequality $(\ref{A12})$ together with (\ref{inv-sym}) will be used to recover the equivalence of  effective Hamiltonians,  which is  the second result of the paper.


\begin{theorem}
\label{TH2}
 The isospectrality 
\begin{equation}
\label{spec123}
{\rm Spec}(  {\rm Op}_\hbar^w (H_1) )  =  {\rm Spec}(  {\rm Op}_\hbar^w (H_2)  )  \quad \quad  \forall \, 0 < \hbar \le 1
\end{equation} 
implies
\begin{equation}
\label{Heff12}
\overline{H}_1  = \overline{H}_2 .
\end{equation}
\end{theorem}
The above result tell us that equality (\ref{Heff12}) is a necessary condition for the isospectrality  (\ref{spec123}). Conversely, the condition  (\ref{spec123}) is a sufficient condition to have the same homogenization  given by (\ref{Heff12}).
\indent In  the  assumption of a quantum integrable system, we have $1 \le i \le n$  commuting semiclassical pseudodifferential operators ${\rm Op}_\hbar (f_i)$. The principal symbols are commuting with respect to the Liouville bracket and  hence we have a (classical) completely integrable Hamiltonian system with momentum map $F (x,p) = (f_1, \, ...,  f_n) (x,p)$. As shown in Thm. 2 of  \cite{PPN}, if the symbols are bounded then  the convex hull of the joint spectra linked to ${\rm Op}_\hbar (f_i)$  recovers, in the classical limit as $\hbar \to  0^+$, the convex hull of the image of the momentum map $F$.  Such a result solves to the isospectrality problem for quantum toric systems with bounded symbols. \\
In our paper, we focus the attention on the effective Hamiltonian  $\overline{H} (P)$ which can be regarded as the generalization, beyond integrability assumptions, of the smooth Hamiltonian depending only from the Action variables $P$ provided by  the Liouville-Arnold Theorem. In the case that $S(P,x)$ solving (\ref{HJ00}) is a smooth and global generating function of a canonical map, then the inversion of $p= P + \nabla_x S (P,x)$ provides the momentum map $F(x,p) := P(x,p)$. 
In this case the image of the momentum map, on the sublevel set for energy $E$,  equals
\begin{equation}
\label{IF}
F ( \{ (x,p) \in \Bbb T^n \times \Bbb R^n \, :  \,  H(x,p) \le E  \}  )  = \{  P \in \Bbb R^n \,  :  \,    \overline{H} (P)  \le E  \}.
\end{equation}
In view of our Theorem \ref{TH2} and equality (\ref{IF}), we recover  (see Section \ref{SEC5})  the same kind of result shown \cite{PPN}, here  in the setting of integrable  Hamiltonian with mechanical form  $\frac{1}{2} |p|^2 + V(x)$ on $\Bbb T^n :=  (\Bbb R / 2 \pi \Bbb Z)^n$.\\
\indent The one dimensional version of Theorem \ref{TH2} is shown in \cite{PZ2}. In particular, it is shown an exact version of the Egorov Theorem for the class of one dimensional periodic Schr\"odinger operators $- \frac{1}{2} \hbar^2 \frac{d^2}{dx^2}  +   V(x)$. We also stress that in this case the effective Hamiltonian is given by the inversion of the map
\begin{equation}
E \mapsto \mathcal{J} (E) := \frac{1}{2\pi} \int_0^{2\pi} \sqrt{2(E-V(x))} \, dx, \quad E \ge \max V,
\end{equation}
namely $\overline{H}(P) = \mathcal{J}^{-1}(P)$ for $|P|\ge (2\pi)^{-1} \int_0^{2\pi} \sqrt{2(\max V-V(x))} \, dx$ and $\overline{H}(P) = \max V$ for $|P|< (2\pi)^{-1} \int_0^{2\pi} \sqrt{2(\max V-V(x))} \, dx$. 
Moreover it is easily seen that 
\begin{equation}
\mathcal{J} (E) = {\rm Vol} \{ (x,p) \in \Bbb T \times \Bbb R \, :  \,  \frac{1}{2}|p|^2   +   V(x) \le E  \}  
\end{equation}
which is the first spectral invariant of the Weyl Law (see \cite{G-S}) on the number of the eigenvalues $(E_{\hbar ,\ell})_{\ell \in \Bbb N}$ smaller than $E$. Namely $N(\hbar,E) := \sharp \{E_{\hbar ,\ell} \, : \,\,  \min V \le E_{h,\ell} \le E   \}$ has the asymptotics  
\begin{equation}
N(\hbar,E)  =   (2\pi h)^{-1} \Bigl(   {\rm Vol}  \{ (x,p) \in \Bbb T \times \Bbb R \,  : \,  \min V \le   \frac{1}{2}|p|^2   +   V(x)  \le E \} +  \mathcal{O} (\hbar)  \Bigr).
\end{equation}
A  study of the link between the effective Hamiltonian $\overline{H}(P)$, viscosity solutions of Hamilton-Jacobi equation and eigenfunctions (or quasimodes)  in the semiclassical framework should involve  their phase space localization. Some preliminary results, beyond the  one dimensional case, have been obtained in \cite{C-Z}, \cite{PZ1}, \cite{Z1}.\\ 
In order to avoid confusion with respect to  the literature, we underline that the effective Hamiltonian approach provided in \cite{G-M-S} is a completely different tool from the  effective Hamiltonian function of the Hamilton-Jacobi  homogeni\-zation used in the present paper.  


\section{Pseudodifferential operators on the flat torus}
\label{SEC2}
In this section we recall  the Weyl quantization on the flat torus as discussed in section 2 of \cite{PZ1}, which make use of the toroidal Pseudodifferential operators theory developed in \cite{R-T}. As shown in Prop. 2.3 of \cite{PZ1}, the  Weyl quantization on the flat torus here applied coincides to the one used in \cite{G-P}.\\ 
Consider the  class of ($x$ - periodic)  H\"ormander's symbols $b \in S^m (\mathbb{T}^n \times \mathbb{R}^n)$, $m \in \mathbb{R}$, consisting of those functions  
in $C^\infty (\mathbb{R}^n \times \mathbb{R}^n)$ which are  $2\pi\mathbb{Z}^n$-periodic in $x$ ($2\pi$-periodic
in each variable $x_j$, $1\leq j\leq n$) and such that  $\forall \alpha, \beta \in \mathbb{Z}_+^n$  there exists $C_{\alpha \beta} >0$ such that 
\begin{equation}
\label{S-def78}
| \partial_x^\alpha\partial_\eta^\beta   b (x,\eta) |   \le  C_{\alpha \beta m} \langle \eta \rangle^{m-|\beta|},  \qquad  \forall (x,\eta)\in \mathbb{T}^n\times\mathbb{R}^n
\end{equation}
where $\langle\eta\rangle:=(1+|\eta|^2)^{1/2}$. Together with these symbols, one can introduce the class of toroidal symbols $\tau \in S^m (\mathbb{T}^n \times \mathbb{Z}^n)$. That is,
by using the difference operator $\Delta_{\kappa_j} \tau (x,k):= \tau (x,\kappa+e_j)-  \tau (x,\kappa)$ ($e_j$ being the $j$th vector of the canonical basis of $\mathbb{R}^n$), we can require that $\forall\alpha,\beta\in\mathbb{Z}^n_+$ there exists $C_{\alpha\beta}$ such that 
\begin{equation}
|\partial_x^\alpha\Delta_\kappa^\beta \tau (x,\kappa)| \leq C_{\alpha \beta  m } \langle\kappa\rangle^{m-|\beta|},\qquad \forall (x,\kappa)\in \mathbb{T}^n\times\mathbb{Z}^n. 
\end{equation}
\begin{remark}
\label{tor-eu}
There is a full correspondence between symbols $S^m (\mathbb{T}^n \times \mathbb{R}^n)$ and $S^m (\mathbb{T}^n \times \mathbb{Z}^n)$. Indeed, we address the reader to Remark 4.1 and Thm. 5.2 of \cite{R-T},  and  we stress that we have a toroidal symbol $\tau \in S^m (\mathbb{T}^n \times \mathbb{Z}^n)$ if and only if there exists $b \in S^m (\mathbb{T}^n \times \mathbb{R}^n)$ such that $\tau = b \, |_{\mathbb{T}^n \times \mathbb{Z}^n}$ and such $b$ is unique modulo a term in $S^{-\infty} (\mathbb{T}^n \times \mathbb{R}^n)$. In particular, since we have the equality  
\begin{equation}
\partial_x^\alpha \Delta_\eta^\beta b (x,\eta) =  \partial_x^\alpha  \partial_\eta^\beta   b (x,  \omega) 
\end{equation}
for some $\omega \in Q := [ \eta_1 + \alpha_1, ... ,  \eta_n + \alpha_n]$ then
\begin{eqnarray}
| \,  \partial_x^\alpha \Delta_\eta^\beta b (x,\eta)  \, |  =  | \,  \partial_x^\alpha  \partial_\eta^\beta   b (x, \omega)  \, |  &\le&  C_{\alpha \beta m} \langle \omega \rangle^{m-|\beta|}
\le   C_{\alpha \beta m}^\prime \langle \eta \rangle^{m-|\beta|}
\end{eqnarray}
with new constants $C_{\alpha \beta m}^\prime > 0$.
\end{remark}

For any  $\tau \in S^m (\mathbb{T}^n \times \mathbb{Z}^n)$ we can thus associate a toroidal Pseudodifferential operator as
\begin{equation}
\mathrm{Op} (\tau) \psi(x) :=(2\pi)^{-n}\sum_{\kappa \in \mathbb{Z}^n}\int_{\mathbb{T}^n}e^{i\langle x-y,\kappa \rangle}  \tau ( x,    \kappa  )\psi(y)dy,\,\,\,\,
\psi\in C^\infty(\mathbb{T}^n).
\end{equation}
The composition rule of these operators (see Thm. 4.3 in \cite{R-T}) states that for any $\tau \in S^m (\mathbb{T}^n \times \mathbb{Z}^n)$ and $w \in S^\ell (\mathbb{T}^n \times \mathbb{Z}^n)$ we have 
\begin{equation}
\mathrm{Op} (\tau) \circ \mathrm{Op} (w) =  \mathrm{Op} (\tau \sharp w)
\end{equation}
with $\tau \sharp w \in S^{m+ \ell} (\mathbb{T}^n \times \mathbb{Z}^n)$ and 
\begin{equation}
\label{comp00}
\tau \sharp w  \sim  \sum_{\alpha \ge 0} \frac{1}{\alpha !}  \Delta_\kappa^\alpha \tau (x,\kappa) \, D_x^{(\alpha)} \overline{w (x,\kappa)}. 
\end{equation}
This asymptotics of symbols  (\ref{comp00}) means that the remainder
\begin{equation}
R_N (x,\kappa) :=  \sum_{\alpha > N} \frac{1}{\alpha !}  \Delta_\kappa^\alpha \tau (x,\kappa) \, D_x^{(\alpha)} \overline{w (x,\kappa)}
\end{equation}
belongs to $S^{m+ \ell} (\mathbb{T}^n \times \mathbb{Z}^n)$ and  $\forall N \ge 1$ fulfills
\begin{equation}
\label{est-RR}
|\partial_x^{\gamma} \Delta_\kappa^{\mu} R_N (x,\kappa)| \leq C_{\gamma \mu  m \ell } \, \langle\kappa\rangle^{m + \ell -|\mu|-N},\qquad \forall (x,\kappa)\in \mathbb{T}^n\times\mathbb{Z}^n. 
\end{equation}
Notice that for symbols of kind $\tau=\tau (x, \hbar \kappa)$ we have that  $\Delta_\kappa^\alpha [ \tau (x,\hbar \kappa)] =  \hbar \, \Delta_\kappa^\alpha \tau (x, \hbar \kappa)$ and this provides series (\ref{comp00}) with powers of $\hbar$. Following the computations (see Thm. 4.3 in \cite{R-T}) of the estimate in (\ref{est-RR}),  we recover in this case the constant $\hbar^N C_{\gamma \mu  m \ell }$. \\

The semiclassical framework of toroidal operators deals with the following
\begin{definition}[Weyl quantization]\\
 Given a symbol $b\in S^m(\mathbb{T}^n\times\mathbb{R}^n)$, we define  the Weyl quantization 
\begin{equation}
\mathrm{Op}^w_\hbar (b) \psi(x) :=(2\pi)^{-n}\sum_{\kappa \in \mathbb{Z}^n}\int_{\mathbb{T}^n}e^{i\langle x-y,\kappa \rangle}  b \Big( y,  \frac{\hbar}{2}  \kappa  \Big)\psi(2y-x)dy,\,\,\,\,
\psi\in C^\infty(\mathbb{T}^n).
\label{eqWeylOurs}
\end{equation}
\end{definition}

\begin{remark}
The link between Weyl quantization and standard quantization is given by 
\begin{equation}
\label{W-St}
\mathrm{Op}^w_\hbar (b) \psi =  \mathrm{Op} (\sigma(\hbar)) \psi
\end{equation}
where $\sigma (\hbar,x,\kappa) =  b(x,\hbar \kappa)  + \mathcal{O}(\hbar)$ in $S^m (\mathbb{T}^n \times \mathbb{Z}^n)$. 
To prove (\ref{W-St}), we observe that $T_\omega \psi(y) := \psi(2y-\omega)$ can be written as
\begin{equation}
T_\omega \psi(y) = (2\pi)^{-n}\sum_{\kappa \in\mathbb{Z}^n}\int_{\mathbb{T}^n}e^{i\langle (2y-\omega) -z,\kappa \rangle} \psi(z)dz, \quad \forall \ \psi \in C^\infty (\mathbb{T}^n),
\end{equation}
and hence (thank to Thm. 4.2 in \cite{R-T}) we have ${\rm Op}^w_\hbar (b) \psi (x) =  (\sigma (X,D) \circ T_{\omega = x} \, \psi )(x)$ with $\sigma \sim \sum_{\alpha\geq 0}\frac{1}{\alpha!}\triangle_\eta^\alpha D_y^{(\alpha)} b(y,\hbar \eta / 2)\bigl|_{y=x}$. By applying Theorem 8.4 of \cite{R-T} it is finally recovered (\ref{W-St}). 
\end{remark}

\medskip

Together with the above notion we can introduce the next
\begin{definition}[Wigner transform]\\
The Wigner transform $W_ \hbar \psi (x,\eta)$ for $x \in \Bbb T^n$ and $\eta\in\frac{\hbar}{2}\mathbb{Z}^n$ is  defined as
\begin{equation}
W_\hbar  \psi(x,\eta) = (2\pi)^{-n}\int_{\mathbb{T}^n} e^{2 i\hbar^{-1} \langle z,\eta\rangle}\psi(x-z)\overline{\psi(x+z)}dz,
\label{WignerT}
\end{equation}
and the Wigner distribution is therefore given by
\begin{equation}
\langle {\rm Op}^w_\hbar (b) \psi , \psi\rangle_{L^2}  =  \sum_{ \eta \in \frac{\hbar}{2} \mathbb{Z}^n}\int_{\mathbb{T}^n} b(x,\eta) W_\hbar \psi(x,\eta)dx. 
\label{WignerD}
\end{equation}
\end{definition}

\begin{remark}
\label{Wig-reg}
We notice that for $(x,\kappa) \in \Bbb T^n \times \Bbb Z^n$, 
\begin{equation}
\partial_x^\alpha \Delta_\kappa^\beta ( W_\hbar  \psi (x,\frac{\hbar}{2} \kappa) ) = (2\pi)^{-n}\int_{\mathbb{T}^n} \Delta_\kappa^\beta e^{i \langle z,\kappa\rangle}  \partial_x^\alpha [\psi(x-z)\overline{\psi(x+z)}] dz.
\end{equation}
In particular (see Prop. 3.1 in \cite{R-T}), $\Delta_\kappa^\beta f ( \kappa ) =  \sum_{\gamma \le \beta} (-1)^{|\beta - \gamma|}   \binom{\beta}{\gamma}  f(\kappa + \gamma)$  so  that  we have $\Delta_\kappa^\beta e^{i \langle z,\kappa\rangle}  =   e^{i \langle z,\kappa \rangle}   \sum_{\gamma \le \beta} (-1)^{|\beta - \gamma|}   \binom{\beta}{\gamma}  e^{i \langle z,\gamma \rangle}  =:  e^{i \langle z,\kappa \rangle}  s(z,\kappa;\beta)$. The above terms reads\\ 
\begin{equation}
\label{pegion}
\partial_x^\alpha \Delta_\kappa^\beta ( W_\hbar  \psi (x,\frac{\hbar}{2} \kappa) ) = (2\pi)^{-n}\int_{\mathbb{T}^n} e^{i \langle z,\kappa \rangle}  s(z,\kappa;\beta)  \partial_x^\alpha [\psi(x-z)\overline{\psi(x+z)}] dz.
\end{equation}
Moreover, we recall that the toroidal Fourier Transform which is $F (\varphi) (\kappa) := (2\pi)^{-n} \int_{\mathbb{T}^n}e^{-i\langle z,\kappa \rangle}  \varphi(z)dz$  maps $C^\infty (\Bbb T^n)$  into 
$\mathcal{S} (\Bbb Z^n)$, i.e. the set of rapidly decaying functions from $\Bbb Z^n$  to $\Bbb C$.  Thus, $|F (\varphi) (\kappa)| \le C_N \langle \kappa \rangle^{-N}$ and applying this inequality to  (\ref{pegion})
\begin{equation}
\label{pegion2}
| \partial_x^\alpha \Delta_\kappa^\beta ( W_\hbar  \psi (x,\frac{\hbar}{2} \kappa) ) |  \le C_{N\alpha \beta} \langle \kappa \rangle^{-N}, \quad \forall N \in \Bbb N.
\end{equation}
This means that $(x,\kappa) \mapsto   W_\hbar  \psi (x,\frac{\hbar}{2} \kappa)$ belongs to $S^m(\mathbb{T}^n\times\mathbb{Z}^n)$ for any $m \in \Bbb Z$. 
\end{remark}

\quad

As in the Euclidean Weyl quantization (see \cite{Mar}, \cite{MZ}) here we have that for any  $\varphi_\hbar \in C^\infty (\Bbb T^n)$ with $\| \varphi_\hbar   \|_{L^2} \le 1$ the projection operator $\pi_\hbar \psi :=  \langle \varphi_\hbar,  \psi \rangle_{L^2}\,  \varphi_\hbar$ can be regarded as a Weyl operator whose symbol is the Wigner transform $W_\hbar  \varphi_\hbar$, see Lemma \ref{Lemma-Weyl}:  
\begin{equation}
\pi_\hbar \psi =   \mathrm{Op}^w_\hbar (q(\hbar)) \psi, \quad q (\hbar,x,\eta) :=   W_\hbar  \varphi_\hbar (x,\eta). 
\label{pro-weyl}
\end{equation}
In view of (\ref{eqWeylOurs}), and by using the properties of the toroidal Fourier Transform $F (\psi) (\kappa) := (2\pi)^{-n} \int_{\mathbb{T}^n}e^{-i\langle y,\kappa \rangle}  \psi(y)dy$ together with its inverse $F^{-1} (g) (x) := \sum_{\kappa \in \mathbb{Z}^n}   e^{i \langle x,\kappa \rangle}  g(\kappa)$ (see section 2 of \cite{R-T})  it is easily seen that 
\begin{equation}
- \frac{1}{2} \hbar^2 \Delta_x + V  =  \mathrm{Op}^w_\hbar (H)
\end{equation}
when $H = \frac{1}{2} |\eta|^2  +  V(y)$. For a given pair  $a \in S^\ell (\mathbb{T}^n \times \mathbb{R}^n)$  and    $b \in S^m (\mathbb{T}^n \times \mathbb{R}^n)$   it is shown in Thm. 2.4 of \cite{TP-Z} that, by using the compostion formula of toroidal Pseudodifferential operators, see (\ref{comp00}), we can write
\begin{equation}
{\rm Op}^w_\hbar (a) \circ  {\rm Op}^w_\hbar (b)  =   {\rm Op}^w_\hbar (a \sharp b ),   \qquad a \sharp b     \sim a \cdot b + \mathcal{O}(\hbar) \in  S^{\ell  + m} (\mathbb{T}^n \times \mathbb{R}^n). 
\end{equation}
Moreover, 
\begin{equation}
[ {\rm Op}^w_\hbar (a) , {\rm Op}^w_\hbar (b)   ] =  {\rm Op}^w_\hbar (a \sharp b  -   b \sharp a) 
\end{equation}
where the so-called Moyal bracket reads 
\begin{equation}
\label{moyal}
\{ a, b \}_M := a \sharp b  -   b \sharp a    \sim - i \hbar  \, \{ a , b \} + \mathcal{O}(\hbar^2)
\end{equation} 
with terms in  $S^{\ell  + m-1} (\mathbb{T}^n \times \mathbb{R}^n)$ and where $ \{ a , b \} $ is the usual Liouville bracket. The difference 
\begin{equation}
\label{moyal-D}
 \{ a, b \}_M - ( - i \hbar  \, \{ a , b \})  
\end{equation} 
belongs to $S^{\ell  + m - 2} (\mathbb{T}^n \times \mathbb{R}^n)$ and exhibits a semiclassical asymptotics with leading order  $ \mathcal{O}(\hbar^2)$.\\

\noindent
In the next we recall the toroidal version of Calderon-Vaillancourt Theorem. 
\begin{theorem}[see \cite{G-P}]
\label{TH-CV}
Let  ${\rm Op}^w_\hbar (b)$ as in  (\ref{eqWeylOurs}) with $b \in S^0 (\mathbb{T}^n \times \mathbb{R}^n)$. Let $M = \frac{1}{2} n+1$ when $n$ is even, $M = \frac{1}{2} (n+1) +1$ when $n$ is odd. Then, for any $\psi \in C^\infty (\mathbb{T}^n)$
\begin{equation}
\|  {\rm Op}_\hbar^w (b)  \psi \|_{L^2(\Bbb T^n)} \le \frac{2^{n+1}}{n+2} \ \frac{  \pi^{\frac{3n-1}2{}}}{\Gamma (\frac{n+1}{2} ) }  \sum_{|\alpha| \le 2M}  \,    \|  \partial_x^\alpha b \|_{L^\infty (\Bbb T^n \times \Bbb R^n)}    \ \| \psi \|_{L^2 (\Bbb T^n)}. 
\end{equation}
\end{theorem}

\section{Egorov Theorem}
\label{SEC3}
We provide the Egorov Theorem in the toroidal setting, namely in the framework of the toroidal Weyl quantization shown in Section \ref{SEC2}. Here we are not interested in the full asymptotics  of the symbols at all orders $\mathcal{O}(\hbar^N)$, since we need only to deal with the leading term involving the Hamiltonian flow.   We mainly follow the same arguments showed in \cite{B-R}, \cite{B-G-P} for the euclidean setting.

\begin{theorem}
\label{TH-Ego}
Let $b \in S^{2} (\mathbb{T}^n \times \mathbb{R}^n)$, $a \in S^{m} (\mathbb{T}^n \times \mathbb{R}^n)$  with  $m \le 0$.  Let $U_\hbar (t) := \exp{ (- i \, {\rm Op}_\hbar^w (b) t / \hbar)}$ and let $\varphi^t : \mathbb{T}^n \times \mathbb{R}^n \to \mathbb{T}^n \times \mathbb{R}^n$ be the Hamiltonian flow of $b$. Then,  
\begin{equation}
U_\hbar^\star (t) \circ   {\rm Op}_\hbar^w  (a)  \circ  U_\hbar  (t)  =   {\rm Op}_\hbar^w  (a_\varphi)
\end{equation} 
where  $a_\varphi \in S^{m} (\mathbb{T}^n \times \mathbb{R}^n) \subset S^{0} (\mathbb{T}^n \times \mathbb{R}^n)$ depends on $\hbar$ and 
\begin{equation}
\|   {\rm Op}_\hbar^w  ( a_\varphi   )  -    {\rm Op}_\hbar^w  (  a \circ \varphi^t  )  \|_{L^2 \to L^2} = \mathcal{O}(\hbar).
\end{equation}
Furthermore,  $a_\varphi  = a \circ \varphi^t   + \mathcal{O}(\hbar)$  in $S^{0} (\mathbb{T}^n \times \mathbb{R}^n)$ uniform in $0 \le t \le 1$, more precisely there exists a family of functionals  $B_{\alpha \beta}[a,b] > 0$  such that for any $0 \le t \le 1$ and  $\forall (x,\eta) \in \Bbb T^n \times \Bbb R^n$ 
\begin{equation}
\label{est-simb80}
 | \partial_x^\alpha  \partial_\eta^\beta (a_\varphi  -   a \circ \varphi^t) (x,\eta)  | \le  B_{\alpha \beta} [a,b] \, \hbar .
\end{equation}

\end{theorem}
\begin{proof}
Let us define $a_0 (t,x,p) :=  a \circ \varphi^t (x,p)$, which be denoted by $a_0 (t)$ in the next. The standard approach in the Euclidean setting (see for example Thm. 1.1 in \cite{B-G-P} or Thm. 1.2 in \cite{B-R}) is to use the functional equality 
\begin{eqnarray}
&& U_\hbar^\star (t) \circ   {\rm Op}_\hbar^w  (a)  \circ  U_\hbar  (t)  -   {\rm Op}_\hbar^w  (a_0 (t))    
\\
&& =   \int_{0}^t      U_\hbar^\star (t-s) \circ \Big(   \frac{i}{\hbar} [  {\rm Op}_\hbar^w (b) ,  {\rm Op}_\hbar^w  (a_0 (s))   ]  -   {\rm Op}_\hbar^w  (\{   b, a_0 (s)  \}  )      \Big) \circ    U_\hbar (t-s)         \, ds
\nonumber
\end{eqnarray}
which can be used also in the toroidal setting. In view of asymptotics  of the Moyal bracket shown in (\ref{moyal}) it follows  
\begin{eqnarray}
  \frac{i}{\hbar}   [  {\rm Op}_\hbar^w (b) ,  {\rm Op}_\hbar^w  (a_0 (s))   ]  -  {\rm Op}_\hbar^w  (\{   b, a_0 (s)  \}  )  =    {\rm Op}_\hbar^w  (  \phi (s)  ) 
\end{eqnarray}
where $\phi (s) :=  \frac{i}{\hbar}  \{ b , a_0 (s) \}_M -  \{   b, a_0 (s)  \}$  belongs to $S^{m+2-2} (\mathbb{T}^n \times \mathbb{R}^n) =  S^{m} (\mathbb{T}^n \times \mathbb{R}^n) \subset S^{0} (\mathbb{T}^n \times \mathbb{R}^n)$, see   (\ref{moyal-D}), and exhibits a semiclassical asymptotics with leading order $\mathcal{O}(\hbar)$. Any term of the semiclassical asymptotics of 
$\phi (s,z) \sim \sum_{\alpha \ge 1} \phi_\alpha (s,z) \hbar^\alpha   =   \phi_1 (s,z) \hbar   +   \mathcal{O}(\hbar^2)   $ in $S^{m} (\mathbb{T}^n \times \mathbb{R}^n)$  depends polinomially from the derivatives of the $C^\infty$ - flow $\varphi^s$,  from the derivatives of $a$ and $b$.     We remind that here we have assumed $m \le 0$, and hence we can apply the Calderon-Vaillancourt Theorem (see Thm. \ref{TH-CV}) to have 
\begin{eqnarray}
\|   {\rm Op}_\hbar^w  (  \phi (s)  )   \|_{L^2 \to L^2}  \le \frac{2^{n+1}}{n+2} \ \frac{  \pi^{\frac{3n-1}2{}}}{\Gamma (\frac{n+1}{2} ) }  \sum_{|\gamma| \le 2M}  \,    \|  \partial_x^\gamma \phi (s, \cdot \,) \|_{L^\infty (\Bbb T^n \times \Bbb R^n)}
\end{eqnarray}
where $M = \frac{1}{2} n+1$ when $n$ is even, $M = \frac{1}{2} (n+1) +1$ when $n$ is odd.\\
Recalling (\ref{S-def78}), Remark \ref{tor-eu} and  (\ref{est-RR}),  any $\partial_x^\gamma \phi$ belongs to $S^{0} (\mathbb{T}^n \times \mathbb{R}^n)$ and can be written as $\partial_x^\gamma \phi_1 (s,z) \hbar   +   \mathcal{O}(\hbar^2)$ with remainder uniform in the variable $s$. Thus, for some $K_\gamma > 0$ we have the estimate
 $ \|  \partial_x^\gamma \phi (s,\cdot \,) \|_{L^\infty} \le K_\gamma \, \hbar$ for any $s \in [0,t]$. \\  
Since $U_\hbar$ are unitary operators, we have $\| U_\hbar \|_{L^2 \to L^2} =  1$, and thus 
\begin{eqnarray}
&& \|   {\rm Op}_\hbar^w  ( a_\varphi   )  -    {\rm Op}_\hbar^w  (  a \circ \varphi^t  )  \|_{L^2 \to L^2} \le \sup_{s \in [0,t]}   \|   {\rm Op}_\hbar^w  (  \phi (s)  )   \|_{L^2 \to L^2} 
\\
&& \le   \sum_{|\gamma| \le 2M}  K_\gamma \, \hbar  =   \mathcal{O}(\hbar).
\end{eqnarray}
The estimate (\ref{est-simb80}) easily follows from the computations  given in Theorem 1.2 of \cite{B-R}  (done in the setting of euclidean Weyl quantization) which works also for toroidal operators and toroidal symbols $S^{m} (\mathbb{T}^n \times \mathbb{Z}^n)$. Then, by recalling Remark \ref{tor-eu} on the correspondence between toroidal symbols  and ($x$-periodic) H\"ormander symbols $S^{m} (\mathbb{T}^n \times \mathbb{R}^n)$ one get the same estimate.
$\Box$
\end{proof}

\begin{corollary}
\label{Cor-aap}
Let $m \le 0$ and $a, b, a_\varphi \in S^{m} (\mathbb{T}^n \times \mathbb{R}^n)$  be as in Theorem \ref{TH-Ego}, and $t \in [0,1]$. Then,  there exists a functional $\mathcal{K}[a,b] > 0$ such that 
\begin{eqnarray}
\| a  -  a_\varphi \|_{C^0 (\Bbb T^n \times \Bbb R^n)} \le  \mathcal{K}[a,b] \, \hbar.   
\end{eqnarray}
\end{corollary}
\begin{proof}
Consider (\ref{est-simb80}) for $\alpha = 0$ and $\beta = 0$ so that we can set  $\mathcal{K}[a,b] := B_{00} [a,b]$. 
\end{proof}


\section{Hamilton-Jacobi equation}
\label{SEC4}
Let $H$ be a Tonelli Hamiltonian, namely $H \in C^2 (\Bbb T^n \times \Bbb R^n  ; \Bbb R)$ is such that the map $p \mapsto H(x,p)$ is convex with positive definite Hessian and $H(x,p) / \|p\| \to + \infty$ as $\|p\| \to + \infty$.\\
For any $P \in \Bbb R^n$, there exists a unique real number $c= \bar{H} (P)$ such that the following cell problem on $\Bbb T^n$:
\begin{equation}
\label{HJ00}
H (x, \, P+\nabla_x S \,) = c, 
\end{equation}
has a  solution $S=S(P,x)$  in the viscosity sense (see \cite{F} and the references therein). As shown in \cite{F}, any viscosity solution is also a weak KAM solution of negative type and belongs to $C^{0,1} (\Bbb T^n)$. Moreover, as shown in \cite{Ri}, any viscosity  solution $S$ exhibits  $C^{1,1}_{{\rm loc}}$ - regularity outside the closure of its singular set $\Sigma (S)$. In particular, $\Bbb T^n \backslash \overline{\Sigma (S)}$ is an open and dense subset of $\Bbb T^n$.
The function $\overline{H}$ is called the effective Hamiltonian, it  is a convex function and  can be represented or approximated in various ways (see for example  \cite{B-S},  \cite{B-S2},  \cite{E2},  \cite{E4},  \cite{E5}). In particular (see  \cite{B-S2} and references therein) we have the following inf-sup formula. Let $v  \in C^{1,1} (\Bbb T^n)$ and $\Gamma :=  \{ (x,\nabla_x v(x)) \in \Bbb T^n   \times \Bbb R^n \ |  \  x \in \Bbb T^n   \}$.  Let $\mathcal{G}$ be the set of all such sets $\Gamma$. 
\begin{equation}
\label{inf-sup}
\overline{H} (P)  =  \inf_{\Gamma \in \mathcal{G}} \sup_{(x,p) \in \Gamma } H(x,p+P).
\end{equation}
The effective Hamiltonian  equals the Mather's $\alpha(H)$  function (see for example \cite{F}, \cite{So}).  The inf-sup formula  can be equivalently computed over $v  \in C^{1} (\Bbb T^n)$ (see for example \cite{E2}) but the Lipschitz regularity of $\nabla v$ has the advantage of being the highest one for which the infimum is a minimum (see \cite{B-S2}).\\ 
As shown in Proposition 1 of \cite{B-S},  for any time one Hamiltonian flows  $\varphi \equiv \varphi^1 : \Bbb T^n \times \Bbb R^n \to \Bbb T^n \times \Bbb R^n$ with $C^1$ - regularity   we have the invariance property 
\begin{equation}
\label{inv-sym22}
\overline{H \circ \varphi }   =  \overline{H} .
\end{equation}
In the mechanical case  $H = \frac{1}{2} |p|^2 + V(x)$  the effective Hamiltonian can be written as
\begin{equation}
\overline{H} (P)  =  \inf_{v \in C^{1,1} (\Bbb T^n)} \sup_{x \in \Bbb T^n}   \frac{1}{2} |  P+ \nabla_x v(x)  |^2   +   V(x) 
\end{equation}
It is easily seen that  $\overline{H} (0) = \inf_{v \in C^{1,1} (\Bbb T^n)} \sup_{x \in \Bbb T^n}   \frac{1}{2} |  \nabla_x v(x)  |^2   +   V(x) = \max V$.\\     
Now let us consider $\{  P \in \Bbb R^n \,  :  \,    \overline{H} (P)  \le E  \} =: \mathcal{U}_E$. In view of the above definition is easy to see that if we denote by $\mathcal{G} (E)$ the set of those 
$\Gamma$ in  $\{ (x,p) \in  \Bbb T^n \times \Bbb R^n \ | \ H(x,P+p)  \le E \}$ then 
\begin{equation}
\label{inf-sup34}
\overline{H} (P)  =  \inf_{\Gamma \in \mathcal{G}(E)} \sup_{(x,p) \in \Gamma } H(x,p+P), \quad \forall P \in  \mathcal{U}_E. 
\end{equation}   

\section{The quantum Integrable case}
\label{SEC5}
For $M = T^\star X$, where $X$ is a manifold of dimension $n$, 
a quantum integrable system is given by $n$  commuting operators ${\rm Op}_\hbar^w (f_i (\hbar))$  with  $f_i : M \to \Bbb R$ and $f_i (\hbar) \simeq f_i + \mathcal{O}(\hbar)$ are such that $\{ f_i, f_j \} = 0$. The  momentum map is given by:  
\begin{equation}
F (x,p) := (f_1, \, ...,  f_n) (x,p).
\end{equation}
As shown in  Thm.  2 of \cite{PPN},  the ``classical  spectrum'' is recovered in the classical limit   as $\hbar \to 0^+$ 
\begin{equation}
\label{CH}
 {\rm Convex \ Hull} ( \, {\rm Joint \ Spec}  ({\rm Op}_\hbar (f_1), ... {\rm Op}_\hbar (f_n)\,  )   \mapsto {\rm Convex \ Hull} \ F (M)  . 
\end{equation}
If we assume that, in our case, the map $(P,x) \mapsto S(P,x)$ solving 
\begin{equation}
H (x,P+\nabla_x S(P,x)) = \overline{H} (P)
\end{equation}
is smooth and  generates  a canonical transformation, then the inversion of $p = P + \nabla_x S (P,x)$ gives the momentum map $F(x,p) := P(x,p)$. Moreover, for any $\Omega_E := \{ (x,p) \in \Bbb T^n \times \Bbb R^n \, :  \,  H(x,p) \le E  \}$  we have  
\begin{eqnarray}
\label{IF}
F(\Omega_E) &=&  \{  P \in \Bbb R^n \,  :  \,    \overline{H} (P)  \le E  \} =: \mathcal{U}_E \, .
\end{eqnarray}

\begin{remark}

\quad
\begin{enumerate}
\item The momentum map $F: \Bbb T^n \times \Bbb R^n \to \Bbb R^n$ does not exists if $H$ is not Liouville Integrable, anyway $\overline{H}$ exists for any Tonelli $H$. 
\item   $\mathcal{U}_E$ is a {\it convex set} (since $\overline{H}$ is a convex function) and  it is {\it spectrally invariant} (thanks to our Theorem \ref{TH2}).
\item   In view of (\ref{CH}) and (\ref{IF}), the set $\mathcal{U}_E$  is a  candidate for a notion of ``{\it classical  spectrum}'' instead of $F(\Omega_E)$, for $H(x,p) = \frac{1}{2} |p|^2  + V(x)$ and beyond the quantum integrable case.
\end{enumerate}
\end{remark}
\noindent
For 1D case, the effective Hamiltonian is given by the inversion of the map
\begin{equation*}
E \mapsto \mathcal{J} (E) := \frac{1}{2\pi} \int_0^{2\pi} \sqrt{2(E-V(x))} \, dx, \quad E \ge \max V,
\end{equation*}
namely 
\begin{displaymath}
\overline{H}(P) = \left\{ \begin{array}{ll}
\max V  & \textrm{if $|P| \le \mathcal{J} (\max V)$,}\\
J^{-1}(P) & \textrm{otherwise.}
\end{array} \right.
\nonumber
\end{displaymath}
Moreover, we stress that $\mathcal{J} (E)$ equals the first spectral invariant of the Weyl Law
\begin{equation*}
\mathcal{J} (E) = {\rm Vol} \{ (x,p) \in \Bbb T \times \Bbb R \, :  \,  \frac{1}{2}|p|^2   +   V(x) \le E  \}  
\end{equation*}
and recall that Bohr-Sommerfeld Rules (see \cite{Ver}, \cite{SN} and references therein) take the form
\begin{equation}
\label{BSV}
\mathcal{S}_{\hbar} ( E_{\hbar,\ell}) =  2\pi \hbar \, \ell   \quad \quad {\rm for} \ \ell = 1,2, ...
\end{equation}
where $E_{\hbar,\ell}$ are the eigenvalues of $- \frac{1}{2} \hbar^2 \Delta_x + V(x)$, and 
\begin{equation}
\label{BSV2}
\mathcal{S}_\hbar ( E )  \sim  \sum_{j=0}^\infty \hbar^j \, {\rm S}_j (E)  =    2\pi \mathcal{J}(E)   + \frac{1}{2} \hbar  \pi  \mu(E)    + \mathcal{O}(\hbar^2).
\end{equation}  
where $\mu(E)$ is the Maslov index of the curve at energy $E$. As shown in Prop. 5.2 of \cite{SN} such a semiclassical series is locally uniform in $E$. 
Thus, the above equalities (\ref{BSV}) - (\ref{BSV2}) imply that two systems with the same Bohr-Sommerfeld Rules have necessarily the same effective Hamiltonian.
Has already shown in \cite{PZ2}, since  $\overline{H} \circ \mathcal{J} (E) = E$ then the  Bohr-Sommerfeld Rules, up to the order  $\mathcal{O}(\hbar^2)$,   implies
\begin{equation}
\label{EH65}
E_{\hbar,\ell} =  \overline{H}( \ell \hbar  -  \mu \hbar / 4 +   \mathcal{O}(\hbar^2) ).
\end{equation}
It is easy to see that, thanks to (\ref{EH65}), the  knowledge of the spectrum of Schr\"odinger operator allow to ``reconstruct'', in the classical limit $\hbar \to 0^+$, the function $\overline{H}$ (see \cite{PZ2}).

\section{Main Results}

\begin{lemma}
\label{Lemma-Weyl}
The  operator $\pi_\hbar \psi :=  \langle \varphi_\hbar,  \psi \rangle_{L^2}\,  \varphi_\hbar$ is a Weyl operator whose symbol is the Wigner transform  of $\varphi_\hbar \in C^\infty (\Bbb T^n)$ where $\|  \varphi_\hbar \|_{L^2} \le 1$ $\forall$ $0 < \hbar \le 1$.  
\end{lemma}
\begin{proof}
We begin by the setting of the Weyl operator with the Wigner transform $W_\hbar \varphi_\hbar$ used as a symbol, 
\begin{eqnarray}
\label{23y}
 \Big({\rm Op}^w_\hbar (  W_\hbar \varphi_\hbar  ) \psi \Big) (x) :=  (2\pi)^{-n}\sum_{\kappa \in \mathbb{Z}^n}\int_{\mathbb{T}^n}e^{i\langle x-y,\kappa \rangle}  W_\hbar \varphi_\hbar \Big( y,  \frac{\hbar}{2}  \kappa  \Big)\psi(2y-x)dy
\nonumber\\ 
\end{eqnarray}
and we recall (see Remark \ref{Wig-reg}) that  the map $(x,\kappa) \mapsto   W_\hbar  \psi (x,\frac{\hbar}{2} \kappa)$ belongs to $S^m(\mathbb{T}^n\times\mathbb{Z}^n)$ for any $m \in \Bbb Z$  (whence for any $m < 0$) and fixed $0 < \hbar \le 1$. In view of this observation, the expression (\ref{23y}) is well posed as a Weyl operator on the torus, and can be rewritten as
\begin{eqnarray}
 (2\pi)^{-2n}  \sum_{\kappa \in \mathbb{Z}^n}  \int_{\mathbb{T}^n}  \int_{\mathbb{T}^n} e^{i\langle x-y + z,\kappa \rangle}  \varphi_\hbar (y-z)\overline{\varphi_\hbar}(y+z) \psi(2y-x)dzdy
\end{eqnarray}
where the two absolutely convergent integrals and the series can be exchanged, 
\begin{eqnarray}
\label{gatto-34}
(2\pi)^{-2n}\int_{\mathbb{T}^n}   \sum_{\kappa \in \mathbb{Z}^n}  e^{i\langle x+ z,\kappa \rangle}  \Big( \int_{\mathbb{T}^n} e^{-i\langle y,\kappa \rangle}   \varphi_\hbar (y-z)\overline{\varphi_\hbar}(y+z)  \psi(2y-x)dy \Big) dz. 
\nonumber\\ 
\end{eqnarray}
The toroidal Fourier transform composed with its inverse is the identity map on $C^\infty (\Bbb T^n)$, and thus (\ref{gatto-34}) becomes
\begin{eqnarray}
(2\pi)^{-n}\int_{\mathbb{T}^n}   \varphi_\hbar (x+z-z)\overline{\varphi_\hbar}(x+z+z)  \psi(2(x+z)-x)   dz,
\end{eqnarray}  
which reads as
\begin{eqnarray}
(2\pi)^{-n}\int_{\mathbb{T}^n} \overline{\varphi_\hbar}(x+2z)  \psi(x+ 2z)   dz  \,  \varphi_\hbar (x).
\nonumber\\
\end{eqnarray}   
Since $T_{x} \psi (z) := \psi (2z+x)$ is linear and preserves scalar product in $L^2 (\Bbb T^n)$,  (write $\psi$ as a Fourier series and easily recover such property) then  we get 
\begin{eqnarray}
(2\pi)^{-n}   \int_{\mathbb{T}^n} \overline{\varphi_\hbar}(y)  \psi(y)   dy  \,  \varphi_\hbar (x) = \langle \varphi_\hbar , \psi \rangle_{L^2}    \,  \varphi_\hbar (x). 
\end{eqnarray}  
$\Box$
\end{proof}

\begin{lemma}
\label{Lemma-red}
Let $V \in C^\infty (\Bbb T^n;\Bbb R)$ and  let $(a,b) \subset \Bbb R_{+}$ so that $a > \min V $. Let $\varphi_{\hbar} \in C^\infty (\Bbb T^n; \Bbb C)$ be an $L^2$-normalized eigenfunction for the operator  
$-\frac{1}{2} \hbar^2 \Delta_x + V(x) : W^{2,2} (\Bbb T^n; \Bbb C) \rightarrow L^{2} (\Bbb T^n; \Bbb C)$ for an eigenvalue $0 < a < E_\hbar < b$ for any $0 < \hbar \le 1$. Let 
$g(\hbar) :=   2 b \  \hbar^{-2} \   e^{1/\hbar}$.  Then, 
\begin{equation}
\varphi_{\hbar}  =  \sum_{k \in \Bbb Z^n ; |k|^2 \le g(\hbar)} \langle \varphi_{\hbar} , e_k \rangle_{L^2} \, e_k     +  r_\hbar 
\end{equation}
where $e_k (x) := \frac{1}{ \sqrt{2\pi}} e^{i k \cdot x}$ and $r_\hbar \in C^\infty (\Bbb T^n;\Bbb C)$ fulfills
\begin{equation}
\label{est-rh}
\|  r_\hbar  \|_{L^2} \le C(b) \,  e^{- 1/ (4\hbar)}, \quad \forall \, 0 < \hbar \le 1.
\end{equation}
where $C(b)$ is given by (\ref{defCb}).
\end{lemma}
\begin{proof}
Since $\varphi_{\hbar} \in L^2 (\Bbb T^n)$ we can write 
$
\varphi_{\hbar}  =  \sum_{k \in \Bbb Z^n} \langle \varphi_{\hbar} , e_k \rangle_{L^2} \, e_k 
$. 
By introducing the cut-off $|k|^2 \le g(\hbar)$ it follows 
\begin{equation}
\varphi_{\hbar} =  \sum_{k \in \Bbb Z^n ; |k| \le g(\hbar)  } \langle \varphi_{\hbar} , e_k \rangle_{L^2} \, e_k   +  r_\hbar
\end{equation}
for some $r_\hbar \in C^\infty (\Bbb T^n)$. In view of the smoothness of the eigenfunction,  $\Delta_x \varphi_{\hbar} \in C^\infty (\Bbb T^n) \subset  L^{2} (\Bbb T^n)$  so that 
$
\Delta_x \varphi_{\hbar}  = - \sum_{k \in \Bbb Z^n} \langle \varphi_{\hbar} , e_k \rangle_{L^2 (\Bbb T^n)} \, |k|^2 e_k 
$.
The eigenvalue equation can be rewritten in terms of the Fourier components, 
\begin{eqnarray}
\Big( \frac{1}{2} \hbar^2 |k|^2 - E_\hbar \Big) \langle \varphi_{\hbar} , e_k \rangle_{L^2 (\Bbb T^n)} = - \langle V \varphi_{\hbar} , e_k \rangle_{L^2},
\end{eqnarray}
and the equality
$
|  \frac{1}{2} \hbar^2 |k|^2 - E_\hbar  | \, | \langle \varphi_{\hbar} , e_k \rangle_{L^2} | =  |\langle V \varphi_{\hbar} , e_k \rangle_{L^2}|
$
gives 
\begin{eqnarray}
| \langle \varphi_{\hbar} , e_k \rangle_{L^2} | \le   \frac{  |\langle V \varphi_{\hbar} , e_k \rangle_{L^2}|  }{\Big |  \frac{1}{2} \hbar^2 |k|^2 - E_\hbar   \Big|} \le   \frac{  \|V\|_{C^0 (\Bbb T^n)}  }{\Big |  \frac{1}{2} \hbar^2 |k|^2 - E_\hbar   \Big|}. 
\end{eqnarray}
Here we look for an estimate of the Fourier components of the remainder, and hence we now consider $|k|^2 > 2 b \, \hbar^{-2} \, e^{1/\hbar}$. Notice that we have 
$\frac{1}{2} \hbar^2 |k|^2 > b \, e^{1/\hbar} \ge b$ $\forall \, 0 < \hbar \le 1$. It follows $\frac{1}{2} \hbar^2 |k|^2 - E_\hbar \ge b - a > 0$ and 
\begin{eqnarray}
| \langle r_{\hbar} , e_k \rangle_{L^2} |  \le   \frac{  \|V\|_{C^0 (\Bbb T^n)}  }{\Big |  \frac{1}{2} \hbar^2 |k|^2 - E_\hbar   \Big|} \le   \frac{  \|V\|_{C^0 (\Bbb T^n)}  }{  \frac{1}{2} \hbar^2 |k|^2 - b } \, .
\end{eqnarray}
We now look for $\bar{C}(b) > 0$  such that 
\begin{eqnarray} 
\frac{  1  }{  \frac{1}{2} \hbar^2 |k|^2 - b } \le   \frac{\bar{C}(b)}{|k|^{3/2}} \, e^{-1/(4\hbar)},
\end{eqnarray}
for any $k \in \Bbb Z^n$ such that $|k|^2 > 2 b \, \hbar^{-2} \, e^{1/\hbar}$. A simple computation shows that we can define 
\begin{eqnarray} 
\label{defh0}
\bar{C}(b)  := 4  (2b)^{-1/2}    \sup_{0 < \hbar \le 1}  \hbar^{-2} \, e^{-1/(2\hbar)}.        
\end{eqnarray}
To conclude,  we have
\begin{equation}
\|  r_\hbar  \|_{L^2}^2 = \sum_{|k|^2 > g(\hbar)}  | \langle \varphi_{\hbar} , e_k \rangle_{L^2} |^2 \le \|V\|_{C^0 (\Bbb T^n)}^2   \bar{C}(b)^2 \hbar^{2N}
 \sum_{k \in \Bbb Z^n, |k|>0}   \frac{1}{|k|^{3}},
\end{equation}
and thus we can define 
\begin{equation}
\label{defCb}
C (b) :=  \|V\|_{C^0 (\Bbb T^n)} \, \bar{C}(b) \, \Big( \sum_{k \in \Bbb Z^n, |k|>0}   \frac{1}{|k|^{3}} \Big)^{1/2}.
\end{equation}
$\Box$
\end{proof}

\begin{remark}
We stress that the hypothesis $a > 0$ used in the statement of the Lemma \ref{Lemma-red} is not restrictive. Indeed, here we look at intervals $(a,b)$  containing only positive values 
$E_\hbar$ of the spectrum. However, by taking a constant $L > \min_{x \in \Bbb T^n} V(x)$ we have that the translated operator $-\frac{1}{2} \hbar^2 \Delta_x + V(x) + L$ has the same eigenfunctions of $-\frac{1}{2} \hbar^2 \Delta_x + V(x)$ and moreover all positive eigenvalues. Thus, the estimate (\ref{est-rh}) for an arbitrary operator works with 
\begin{equation}
C (b) :=  \|V + L \|_{C^0 (\Bbb T^n)} \, \bar{C}(b) \, \Big( \sum_{k \in \Bbb Z^n}   \frac{1}{|k|^{3}} \Big)^{1/2}. 
\end{equation} 
\end{remark}

\begin{proposition}
\label{Prop10}
Let $(a,b) \subset \Bbb R_{+}$ and let $E_{\hbar,\alpha}$ with $\alpha = 1, ...,  N (\hbar,a,b)$ be all the eigenvalues $E_{\hbar,\alpha}$ of $-\frac{1}{2} \hbar^2 \Delta_x + V(x)$ inside $(a,b) \subset \Bbb R_{+}$ repeated with their multiplicity. Let $\varphi_{\hbar,\alpha} \in C^\infty (\Bbb T^n)$ be related eigenfunctions, $L^2$ - normalized and linearly independent. Then, 
\begin{eqnarray}
\Pi_{\hbar} \varphi &:=& \sum_{\alpha = 1}^{N (\hbar,a,b)}   E_{\hbar,\alpha}  \langle \varphi,  \varphi_{\hbar,\alpha} \rangle_{L^2} \  \varphi_{\hbar,\alpha} 
\\
&=& \sum_{|k|^2, |\mu|^2 \le g(\hbar)} \omega_{\hbar,k,\mu} \, \langle   \varphi  ,  e_\mu \rangle_{L^2} \, e_k     +  R_\hbar  \varphi ,
\label{finite-r}
\end{eqnarray}
where $\omega_{\hbar,k,\mu}  :=  \sum_\alpha  E_{\hbar,\alpha}  \, \langle  \varphi_{\hbar,\alpha} , e_\mu \rangle_{L^2}  \,   \langle  \varphi_{\hbar,\alpha} ,  e_k \rangle_{L^2}$  and  
\begin{equation}
\|  R_\hbar  \|_{L^2 \to L^2} \le K(b) \, \hbar^{-n} e^{-1/(4\hbar )}   , \quad  \quad \forall \, 0 < \hbar \le 1,
\end{equation}
where $K(b) := 2 \, b \, \bar{K} (b) \, C(b)$,    $\bar{K} (b) :=  \sup_{0 < \hbar \le 1}  N(a,b,\hbar) \, \hbar^{n} < + \infty$, $C(b)$ is given in (\ref{defCb}).
\end{proposition}
\begin{proof}
Let $G(\hbar) \varphi :=  \sum_{|k| \le g(\hbar)} \langle \varphi , e_k \rangle_{L^2} \, e_k$. In view of Lemma \ref{Lemma-red},
\begin{eqnarray}
\Pi_{\hbar} \varphi &=& \sum_{\alpha = 1}^{N (\hbar,a,b)}   E_{\hbar,\alpha}  \langle \varphi,  \varphi_{\hbar,\alpha} \rangle_{L^2} \  \varphi_{\hbar,\alpha} 
\\
&=& \sum_{\alpha = 1}^{N (\hbar,a,b)}   E_{\hbar,\alpha}  \langle \varphi, G(\hbar) \varphi_{\hbar,\alpha}   + r_{\hbar,\alpha} \rangle_{L^2}   (   G(\hbar) \varphi_{\hbar,\alpha}   + r_{\hbar,\alpha} )
\\
&=& \sum_{\alpha = 1}^{N (\hbar,a,b)}   E_{\hbar,\alpha}  \langle \varphi, G(\hbar) \varphi_{\hbar,\alpha}    \rangle_{L^2}      G(\hbar) \varphi_{\hbar,\alpha} + E_{\hbar,\alpha}  \langle \varphi, G(\hbar) \varphi_{\hbar,\alpha}    \rangle_{L^2}     r_{\hbar,\alpha}
\nonumber
\\
&+& \sum_{\alpha = 1}^{N (\hbar,a,b)}   E_{\hbar,\alpha}  \langle \varphi, r_{\hbar,\alpha} \rangle_{L^2}   (   G(\hbar) \varphi_{\hbar,\alpha}   + r_{\hbar,\alpha} ).
\end{eqnarray}
The leading terms equals 
\begin{eqnarray}
\sum_{\alpha = 1}^{N (\hbar,a,b)}   E_{\hbar,\alpha}  \langle \varphi, G(\hbar) \varphi_{\hbar,\alpha}    \rangle_{L^2}      G(\hbar) \varphi_{\hbar,\alpha} = \sum_{|k|^2, |\mu|^2 \le g(\hbar)} \omega_{\hbar,k,\mu} \, \langle   \varphi  ,  e_\mu \rangle_{L^2} \, e_k 
\end{eqnarray}
where $\omega_{\hbar,k,\mu}  :=  \sum_\alpha  E_{\hbar,\alpha}  \, \langle  \varphi_{\hbar,\alpha} , e_\mu \rangle_{L^2}  \,   \langle  \varphi_{\hbar,\alpha} ,  e_k \rangle_{L^2}$ . The remainder 
\begin{eqnarray}
R_\hbar \varphi =    \sum_{\alpha = 1}^{N (\hbar,a,b)}   E_{\hbar,\alpha} \Big(   \langle \varphi, G(\hbar) \varphi_{\hbar,\alpha}    \rangle_{L^2}     r_{\hbar,\alpha}   + \langle \varphi, r_{\hbar,\alpha} \rangle_{L^2}  \varphi_{\hbar,\alpha}    \Big)
\end{eqnarray}
has the estimate
\begin{eqnarray}
\| R_\hbar \varphi \|_{L^2} &\le&   N (\hbar,a,b)   \sup_{\alpha} |E_{\hbar,\alpha}|  \Big(   \| \varphi  \|_{L^2} \, \| G(\hbar) \varphi_{\hbar,\alpha}\|_{L^2} \, \|   r_{\hbar,\alpha} \|_{L^2}   
\\
&&  \quad \quad \quad \quad \quad  \quad  \quad  \quad \quad  + \,  \| \varphi \|_{L^2} \| r_{\hbar,\alpha} \|_{L^2}  \| \varphi_{\hbar,\alpha}  \|_{L^2}\Big). 
\end{eqnarray}
Notice that $\|  \varphi_{\hbar,\alpha}\|_{L^2} = 1$, $\| G(\hbar) \varphi_{\hbar,\alpha}\|_{L^2} \le \|  \varphi_{\hbar,\alpha}\|_{L^2}  = 1$, $|E_{\hbar,\alpha}| \le b$ and by Lemma \ref{Lemma-red} we have  $\| r_{\hbar,\alpha} \|_{L^2} \le C(b) e^{-1/(4\hbar )}$.\\ In view of the Weyl Law of eigenvalues (see for example \cite{G-S}) applied for the operator $-\frac{1}{2} \hbar^2 \Delta_x + V(x)$ we have that 
\begin{eqnarray}
N (\hbar,a,b) = (2\pi \hbar)^{-n} [ {\rm Vol}( a < \frac{1}{2}|p|^2 + V < b) + \mathcal{O}(\hbar) ].
\end{eqnarray}
This implies that  $\bar{K} (b) :=  \sup_{0 < \hbar \le 1}  N(a,b,\hbar) \, \hbar^{n} < + \infty$. We conclude that 
\begin{eqnarray}
\| R_\hbar \varphi \|_{L^2} \le 2b \, \bar{K} (b)    \, C(b) \, \hbar^{-n} e^{-1/(4\hbar )}  \,  \|  \varphi \|_{L^2} .
\end{eqnarray}
$\Box$
\end{proof}

\bigskip

\noindent
{\bf Proof of Theorem 1}. We begin by recalling that  $-\frac{1}{2} \hbar^2 \Delta_x + V_1(x)$ and  $-\frac{1}{2} \hbar^2 \Delta_x + V_2 (x)$ defined on the flat torus, i.e. with domain $W^{2,2} (\Bbb T^n)$, both exhibit discrete spectrum.  Since we are assuming that these operators have the same spectrum, then this ensures that there exists a unitary operator $U_\hbar : W^{2,2} (\Bbb T^n) \to W^{2,2} (\Bbb T^n)$ such that 
\begin{eqnarray}
U_\hbar^\star \circ (-\frac{1}{2} \hbar^2 \Delta_x + V_1(x)) \circ  U_\hbar = -\frac{1}{2} \hbar^2 \Delta_x + V_2 (x)
\end{eqnarray}
on the domain $W^{2,2} (\Bbb T^n)$. In particular, $U_\hbar  \, \varphi_{\hbar,\alpha}^{(2)} =  \varphi_{\hbar,\alpha}^{(1)}$ for all the eigenfunctions of the two operators. In fact, we can localize such a unitary equivalence in a bounded subset of the spectrum. Namely, for $\Pi_\hbar$ as in Prop. \ref{Prop10} we have   
\begin{eqnarray}
U_\hbar^\star \circ \Pi_\hbar^{(1)} \circ  U_\hbar = \Pi_\hbar^{(2)}. 
\end{eqnarray}
We observe that  (in view of its definition) $\Pi_\hbar$ are finite rank operators, and moreover thanks to (\ref{finite-r}) such a finite rank can be regarded with respect to the ortonormal set $e_k (x)$. Thus, for $|k|^2, |\mu|^2 \le g(\hbar)$ we define the finite dimensional matrix $\mathcal{U}_\hbar (k,\mu) :=  \langle e_k ,  U_\hbar e_\mu \rangle$  and $\mathcal{P}_\hbar (k,\mu) :=  \langle e_k ,  \Pi_\hbar e_\mu \rangle$ and realize that
\begin{eqnarray}
\mathcal{U}_\hbar^\star \circ \mathcal{P}_\hbar^{(1)} \circ  \mathcal{U}_\hbar = \mathcal{P}_\hbar^{(2)}   +   \mathcal{O}(\hbar^\infty). 
\end{eqnarray}
In fact, 
\begin{eqnarray}
\mathcal{P}_\hbar (k,\mu)  &=&  \langle e_k , (-\frac{1}{2} \hbar^2 \Delta_x + V(x))  e_\mu \rangle    +  \mathcal{O}(\hbar^\infty)
\\
&=& \frac{1}{2} \hbar^2  |\mu|^2 \delta_{k \mu}  +    \langle e_k , V(x)  e_\mu \rangle    +  \mathcal{O}(\hbar^\infty). 
\end{eqnarray}
Notice the polynomial (hence $C^\infty$ - type) behavior of the above leading term $\mathcal{P}_{\hbar,0} (k,\mu) := \frac{1}{2} \hbar^2  |\mu|^2 \delta_{k \mu}  +    \langle e_k , V(x)  e_\mu \rangle$. We stress again that  $|k|^2, |\mu|^2 \le g(\hbar)$ and that $g(\hbar) \to + \infty$ as $\hbar \to 0^+$; this is the reason why we cannot find the eigenvalues of   $-\frac{1}{2} \hbar^2 \Delta_x + V(x)$ simply by the eigenvalues of  $\mathcal{P}_\hbar$ with $|k|^2, |\mu|^2 \le L$ for some $\hbar$-independent constant $L > 0$. Anyway, any components of the eigenfunctions and all the eigenvalues of the matrix $\mathcal{P}_{\hbar,0}$ have a $C^0$ - dependence from $\hbar$.\\ 
We then rewrite, thanks to a unitary operator $\mathcal{U}_{\hbar,0}$ the equality
\begin{eqnarray}
\label{uni-P}
\mathcal{U}_{\hbar,0}^\star \circ \mathcal{P}_{\hbar,0}^{(1)} \circ  \mathcal{U}_{\hbar,0} = \mathcal{P}_{\hbar,0}^{(2)}   +   \mathcal{O}(\hbar^\infty). 
\end{eqnarray}
which maps to eigenfunctions of $\mathcal{P}_{\hbar,0}^{(2)}$ into the eigenfunctions of $\mathcal{P}_{\hbar,0}^{(1)}$.  We are now in the position to say that $\mathcal{U}_{\hbar,0}$ is a finite dimensional unitary operator; and that its dependence from $\hbar$ is continuous. 
As a consequence, there exists a (finite dimensional) selfadjoint matrix $\mathcal{A}_\hbar$ such that 
\begin{eqnarray}
\mathcal{U}_{\hbar,0}  =   e^{-i  \mathcal{A}_\hbar} 
\end{eqnarray}
where $\mathcal{A}_\hbar$ has a continuous dependence from $\hbar$. Hence, there is a selfadjoint finite rank operator $A_\hbar$ on the vector space of the functions in $L^2 (\Bbb T^n)$ such that 
\begin{eqnarray}
\varphi (x) := \sum_{|k| \le g(\hbar)} c_k \, e_k (x), \quad   |c_k| \le 1, 
\end{eqnarray}   
such that
\begin{eqnarray}
U_{\hbar}  \varphi   =   e^{-i  A_\hbar}  \varphi   +   \mathcal{O}(\hbar^\infty).
\end{eqnarray}
The entries of the matrix $\langle e_k , A_\hbar e_\mu \rangle$ have continuous dependence from $\hbar$.  We stress that, in view of Lemma \ref{Lemma-red}, the set of all those  
\begin{eqnarray}
\psi (x) := \sum_{\alpha=1}^{N(\hbar,a,E)} d_\alpha \, \varphi_{\hbar,\alpha} (x), \quad   |d_\alpha| \le 1, 
\end{eqnarray}  
can always be written as $\psi  = \varphi  +  \mathcal{O}(\hbar^\infty)$.\\ 
By defining $B_\hbar := \hbar \, A_\hbar$, and recalling Lemma \ref{Lemma-Weyl} we can say that  
\begin{eqnarray}
\label{U-rep}
U_{\hbar} \psi =   e^{-\frac{i}{\hbar}  B_\hbar}  \psi  +   \mathcal{O}(\hbar^\infty) = e^{- \frac{i}{\hbar}  {\rm Op}_\hbar^w (b(\hbar))    }  \psi  +   \mathcal{O}(\hbar^\infty) 
\end{eqnarray}
where $b(\hbar,x,\eta) := \hbar \, \sum_{\alpha=1}^{N(\hbar,a,E)} E_{\hbar,\alpha} W_\hbar \widetilde{\varphi}_{\hbar,\alpha} (x,\eta) \cdot \mathcal{X}_{ \le E + \epsilon } (x,\eta)$ and $\mathcal{X}_{ \le E + \epsilon  } (x,\eta)$ is a $C^\infty$ - compactly supported  function which equals $1$ on $\{ (x,\eta) \in \Bbb T^n \times \Bbb R^n \  |  \  H(x,\eta) \le E  \}$, and equals  equals $0$ on $\{ (x,\eta) \in \Bbb T^n \times \Bbb R^n \  |  \  H(x,\eta) > E + \epsilon   \}$. The functions  $\widetilde{\varphi}_{\hbar,\alpha}$ provide a complete ortonormal set of the eigenfunctions related to $A_\hbar$.  \\
Notice that, for any fixed  value of $0 < \hbar < 1$, the map $z \mapsto b(\hbar,z)$ is smooth, compactly supported into the same set $\{ (x,\eta) \in \Bbb T^n \times \Bbb R^n \  |  \  H(x,\eta) \le E + \epsilon  \}$ and that $\hbar \mapsto b(\hbar,z)$ is continuous. Moreover, 
\begin{eqnarray}
b(\hbar,x, \frac{2}{\hbar}\kappa) = e_{-\kappa} (x) \cdot (  B_{\hbar} e_\kappa  )(x) =   e_{-\kappa} (x) \cdot \Big( \sum_{|\mu|^2 \le g(\hbar)}  \langle e_\mu , B_{\hbar} e_\kappa \rangle e_\mu (x) \Big)
\nonumber
\\
\end{eqnarray}
Thus, for $0 < \hbar \le \sigma \le 1$, 
\begin{eqnarray}
\lim_{ \hbar \to \sigma^- } \|  b(\hbar, \cdot \, ) -   b(\sigma, \cdot  \,) \|_{L^\infty}    =  0.
\end{eqnarray}
This ensures that 
\begin{eqnarray}
\label{Lim-sigma}
\lim_{ \hbar \to \sigma^- } {\rm Op}_\hbar^w (b(\sigma)) \psi -  {\rm Op}_\hbar^w (b(\hbar)) \psi   =  0.
\end{eqnarray}
In view of (\ref{uni-P}) and (\ref{U-rep}), we recover
\begin{eqnarray}
e^{\frac{i}{\hbar}  {\rm Op}_\hbar^w (b(\hbar))    }  \circ   {\rm Op}_\hbar^w (H_1 \, \sharp \, \mathcal{X}_{ \le E + \epsilon  } )   \circ   e^{- \frac{i}{\hbar}  {\rm Op}_\hbar^w (b(\hbar))    } \psi =  
 {\rm Op}_\hbar^w (H_1 \, \sharp \, \mathcal{X}_{ \le E + \epsilon  } ) \,  \psi  .
\nonumber \\
\end{eqnarray}
The limit (\ref{Lim-sigma}) allows to recover, up to the remainder  $\delta(\hbar,\sigma)$, 
\begin{eqnarray}
e^{\frac{i}{\hbar}  {\rm Op}_\hbar^w (b(\sigma))    }  \circ   {\rm Op}_\hbar^w (H_1 \, \sharp \, \mathcal{X}_{ \le E + \epsilon  } )   \circ   e^{- \frac{i}{\hbar}  {\rm Op}_\hbar^w (b(\sigma))    } \psi 
&=&   {\rm Op}_\hbar^w (H_2 \, \sharp \, \mathcal{X}_{ \le E + \epsilon  } ) \,  \psi  
\nonumber
\\
& + & \delta(\hbar,\sigma)  .
\end{eqnarray}
In particular, 
\begin{eqnarray}
\lim_{ \hbar \to \sigma^- } \|  \delta(\hbar,\sigma)  \|_{L^2} =  0.
\end{eqnarray}
The application of Corollary \ref{Cor-aap} involves the  constant $\mathcal{K}[H_1, b(\sigma) ]$ (possibly diverging as $\sigma \to 0^+$), and $\varphi_\sigma$ corresponding to the Hamiltonian flow of $b(\sigma, z)$. This gives
\begin{eqnarray}
\|   H_1 \circ \varphi_\sigma -  H_2 \|_{C^0 (\Omega (E))}   \le  \mathcal{K}[H_1, b(\sigma) ] \, \hbar  
\end{eqnarray}
The  inequality we require is written for arbitrary $0 < u \le 1$ and 
\begin{eqnarray}
  \mathcal{K}[H_1, b(\sigma) ]   \, \hbar  \le u 
\end{eqnarray}
which is fullfilled for $\hbar \le u \cdot  \mathcal{K}[H_1, b(\sigma) ]^{-1}$. Notice that 
\begin{eqnarray}
 u \cdot  \mathcal{K}[H_1, b(\sigma) ]^{-1} \le \sigma 
\end{eqnarray}
is fulfilled for 
\begin{eqnarray}
u \le \sigma \cdot \mathcal{K}[H_1, b(\sigma) ]. 
\end{eqnarray} 
In particular, if $\sigma \to 0$ and $u \to 0$ then $\hbar \to 0$ and $\sigma - \hbar \le \sigma \to 0$. We conclude that for any $0 < \varepsilon \le 1$ there exists an interval $0 < \sigma \le \sigma_0 (E,\varepsilon)$ such that 
\begin{eqnarray}
\|   H_1 \circ \varphi_\sigma -  H_2 \|_{C^0 (\Omega (E))}   \le \varepsilon.
\end{eqnarray} 
\noindent
$\square$

\bigskip
\bigskip

\noindent
{\bf Proof of Theorem 2}. We apply the main result of Theorem \ref{TH1},  combined with (\ref{inv-sym22}) and (\ref{inf-sup34}). More precisely, 

\begin{equation}
\label{inf-sup345}
\overline{H}_2 (P)  =  \inf_{\Gamma \in \mathcal{G}_2(E)} \sup_{(x,p) \in \Gamma } H_2 (x,p+P), \quad \forall P \in  \mathcal{U}_{2,E}. 
\end{equation}
The application of 
\begin{eqnarray}
\|   H_1 \circ \varphi_\sigma -  H_2 \|_{C^0 (\Omega (E))}   \le \varepsilon, \quad 0 < \sigma \le \sigma_0 (E,\varepsilon)
\end{eqnarray} 
ensures that 
\begin{equation}
\overline{H}_2 (P)  =  \inf_{\Gamma \in \mathcal{G}_2(E)} \sup_{(x,p) \in \Gamma } H_1 \circ \varphi_\sigma (x,p+P), \quad \forall P \in  \mathcal{U}_{2,E}. 
\end{equation}
On the other hand,
\begin{equation}
\overline{H}_1 (P)  =  \inf_{\Gamma \in \mathcal{G}} \sup_{(x,p) \in \Gamma } H_1 \circ \varphi_\sigma (x,p+P), \quad \forall P \in \Bbb R^n, \quad \forall \varphi_\sigma,   
\end{equation}
As a consequence,
\begin{equation}
\overline{H}_1 (P)  = \overline{H}_2 (P) , \quad \forall P \in \mathcal{U}_{2,E} 
\end{equation}
and since $E$ can be fixed arbitrary large  then
\begin{equation}
\overline{H}_1 (P)  = \overline{H}_2 (P) , \quad \forall P \in \Bbb R^n.
\end{equation}
$\Box$

\bigskip
\bigskip

\noindent
{\small {\bf Acknowledgments}:  We are grateful to S. Graffi, A. Parmeggiani  and T. Paul for the  many useful discussions  on  the  Egorov Theorem and inverse spectral problem; we  thank F. Cardin and A. Sorrentino  for our enlightening  conversations on the most recent results of Hamilton-Jacobi equation and the inverse homogenization problem.}\\

\noindent
{\small {\bf Funding}: This work has been supported and financed by the Italian National Group of Mathematical Physics (INDAM-GNFM) within the research project 2016/2017: ``Periodic Schr\"odinger operators and weak KAM theory''.}


%




\end{document}